\documentclass{llncs}

\fontfamily{cmr} 

\usepackage{amsmath,amsfonts,graphicx,layout,mathrsfs,algorithm,algorithmic,amssymb,fullpage}

\newcommand{\dsum}[1]{\displaystyle{\sum_{#1}}} 
\newcommand{\dprod}[1]{\displaystyle{\prod_{#1}}}
\newcommand{\dmax}[1]{\displaystyle{\max_{#1}}}
\newcommand{\dmin}[1]{\displaystyle{\min_{#1}}}
\newcommand{\dquote}[1]{``\nolinebreak{#1}\nolinebreak''}
\newcommand{\defeq}{\stackrel{\rm def}{=}}
\newcommand{\perm}[2]{{}_{#1}{\rm P}_{#2}}

\newcommand{\refsec}[1]{Section \nolinebreak\ref{sec:#1}}
\newcommand{\reffig}[1]{Figure \nolinebreak \ref{fig:#1}}
\newcommand{\reftbl}[1]{Table \nolinebreak \ref{tbl:#1}}
\newcommand{\refeq}[1]{\text{Equation} \nolinebreak (\ref{eq:#1})}
\newcommand{\reffm}[1]{\text{Formula} \nolinebreak (\ref{fm:#1})}
\newcommand{\reflm}[1]{\text{Lemma} \nolinebreak \ref{lm:#1}}
\newcommand{\refcor}[1]{\text{Corollary} \nolinebreak \ref{cor:#1}}
\newcommand{\refth}[1]{\text{Theorem} \nolinebreak \ref{th:#1}}
\newcommand{\refalg}[1]{\text{Algorithm} \nolinebreak \ref{alg:#1}}
\newcommand{\refapp}[1]{Appendix \nolinebreak\ref{app:#1}}

\newcommand{\pr}{{\rm Pr}}

\newcommand{\im}{{\rm Im}}

\newcommand{\disp}{\displaystyle}
\newcommand{\dhead}[1]{\vspace{2mm}\indent $\displaystyle{{#1}}$\\}   
\newcommand{\dcont}[1]{\indent \indent $\displaystyle{{#1}}$\\}   
\newcommand{\fright}[1]{\begin{flushright}{#1}\end{flushright}}

\newcommand{\dcontx}[2]{\indent \indent $\displaystyle{{#1}}$\fright{(#2)}}

\newtheorem{atheorem}{Theorem}

\newcommand{\st}{\STATE}

\newcommand{\ignore}[1]{{}}

\newcommand{\kkk}{k}
\begin{document}

\title{
{\it k}-anonymous Microdata Release
via Post Randomisation Method}
\author{Dai Ikarashi \and Ryo Kikuchi \and Koji Chida \and  Katsumi Takahashi}
\institute{
	NTT Secure Platform Laboratories, \\
	\url{{\tt\{ikarashi.dai, kikuchi.ryo, chida.koji, takahashi.katsumi\}@lab.ntt.co.jp}}
}
\normalsize
\maketitle
\begin{abstract}
The problem of the release of anonymized microdata is an important topic
in the fields of statistical disclosure control (SDC) and privacy preserving data publishing (PPDP), 
and yet it remains sufficiently unsolved.
In these research fields, 
$k$-anonymity has been widely studied as an anonymity notion for mainly deterministic anonymization algorithms,
and some probabilistic relaxations have been developed.
However, 
they are not sufficient due to their limitations, i.e., being weaker than the original $k$-anonymity or  
requiring strong parametric assumptions.
First we propose $Pk$-anonymity, 
a new probabilistic $k$-anonymity, and prove that $Pk$-anonymity is a mathematical extension of $k$-anonymity rather than a relaxation.  
Furthermore, $Pk$-anonymity requires no parametric assumptions.\\
This property has a significant meaning in the viewpoint that it enables us to compare privacy levels of probabilistic microdata release algorithms with deterministic ones.  
Second, we apply $Pk$-anonymity to
the post randomization method (PRAM), 
which is an SDC algorithm based on randomization.
PRAM is proven to satisfy $Pk$-anonymity in a controlled way, i.e, 
one can control PRAM's parameter so that $Pk$-anonymity is satisfied.
On the other hand, PRAM is also known to satisfy $\varepsilon$-differential privacy, 
a recent popular and strong privacy notion. 
This fact means that our results significantly enhance PRAM since it implies the satisfaction of {\em both} important notions:
$k$-anonymity and $\varepsilon$-differential privacy.
\end{abstract}


\section{Introduction}


Releasing microdata while preserving privacy has been widely studied in the fields of statistical disclosure control (SDC) and privacy preserving data publishing (PPDP).
Microdata has significant value, 
especially for data analysts who wish to conduct various type of analyses involving the viewing of 
whole data and determining what type of analysis they should conduct. 

The most common privacy notion for microdata release is {\em $k$-anonymity} proposed by Samarati and Sweeney \cite{SS98,Sw02}.
It means that \dquote{no one can narrow down a person's record to $k$ records.}
This semantics is quite simple and intuitive.
Therefore, 
many studies have been conducted on $k$-anonymity, 
and many relevant privacy notions such as $\ell$-diversity \cite{MGK06}, have also been proposed.
Among these relevant studies,
applying $k$-anonymity to probabilistic algorithms is a significant research direction.
Most $k$-anonymization algorithms deterministically generalize or partition microdata. 
However, 
there are probabilistic SDC methods such as random swapping, random sampling, and post randomization method (PRAM) \cite{KWG97}.
How are these probabilistic algorithms related to $k$-anonymity?

Regarding random swapping, 
for example, 
Soria-Comas and Domingo-Ferrer answered the above question by relaxing $k$-anonymity to a probabilistic $k$-anonymity, 
which means that \dquote{no one can {\em correctly link} a person to a record with a higher probability than $1/k$ \cite{SD12}.} 
Intuitively, this semantics seems to be very close to that of the original $k$-anonymity.
However, 
its precise relation to $k$-anonymity has not been argued, and we still cannot definitely say that an algorithm satisfying their probabilistic $k$-anonymity also is $k$-anonymous.

PRAM was proposed by Kooiman et al. in 1997. 
It changes data into other random data according to the probability on a {\it transition probability matrix}. 
Agrawal et al. also developed privacy preserving OLAP (Online Analytical Processing) \cite{AST05} by retention-replacement perturbation, 
which is an instantiation of PRAM. 
For many years, 
PRAM's privacy was not clarified;
however, 
PRAM has been recently proven to satisfy {\em $\varepsilon$-differential privacy (DP)} \cite{LWR12}. 

Differential privacy \cite{Dw06} is another privacy notion that has attracted a great deal of attention recently. 
$\varepsilon$-DP is the original version of DP and many other relevant notions have been developed, 
e.g., $(\varepsilon, \delta)$-DP, 
which is a relaxation of $\varepsilon$-DP.



\subsection{Motivations}

After the proposal, $\varepsilon$-DP has been widely researched and is now known to be very strong privacy notion. 
Thus, it is natural that the satisfaction of $\varepsilon$-DP is important. 
However, especially in the PPDP field, $k$-anonymity is as important as $\varepsilon$-DP, 
although it takes only re-identification into consideration
and several papers showed the limitation of $k$-anonymity \cite{MGK06,LLV07}. 
This notion is very simple and intuitive;
therefore, the enormous number of techniques has been invented, 
and as a result, $k$-anonymity has already spread among the businesspeople, doctors, etc., 
who are conscious about privacy, not only among the researchers. 
From the viewpoint of practice, it is a great merit that people recognize and understand the notion.

Therefore, merging the two notions while preserving 
their theoretical guarantees in a controlled way is desirable.
However, 
$k$-anonymity applies only to deterministic anonymization algorithms,
and $\varepsilon$-DP applies to randomized ones;
thus, it has been hard to manage both of them at once until now.

PRAM has several good features, and 
we believe that it is one of promising candidates for PPDP.
The anonymization step in PRAM is performed by a record-wise fashion
so anonymizing data in parallel is easy, and 
we can extend PRAM to a local perturbation, i.e., 
an individual anonymizes his/her data before sending them to the central server.
In addition, PRAM does not needs generalization,
so we can obtain anonymized data with fine granularity 
and perform a fine-grained analysis on them.
Furthermore, it is known that PRAM can satisfy 
$\varepsilon$-DP \cite{LWR12}.

Although PRAM has these features
and was proposed \cite{KWG97} before when the methods 
satisfying $k$-anonymization \cite{Sw02} and 
satisfying $\varepsilon$-DP \cite{Dw06} were proposed,
it has been studied less than other approaches
in the area of PPDP.
Most popular methods for PPDP are evaluated in the context of $k$-anonymity.
However, PRAM is a probabilistic method,
so it cannot be evaluated in the context of $k$-anonymity.
This means that no one can compare PRAM with other methods for PPDP
in the same measure.

From the above circumstances, our aim of the paper is twofold.
First, we extend $k$-anonymity for probabilistic methods
(not only PRAM) for merging $k$-anonymity and $\varepsilon$-DP.
Second, we evaluate how strongly PRAM preserves privacy
in the context of $k$-anonymity.

\subsection{Contributions}

Our contributions are the following two points. 

\paragraph{Extending $k$-anonymity for Probabilistic Methods}
We propose $Pk$-anonymity, which
has the following four advantages compared to current probabilistic $k$-anonymity notions. 

1. It is formally defined and sufficient to prove that 
it is a rigorous extension of the original $k$-anonymity. 
Specifically, we prove that $k$-anonymity and $Pk$-anonymity are totally equivalent if an anonymization algorithm is deterministic, 
in other words, 
if the algorithm is in the extent of conventional $k$-anonymization.
We claim that 
one can consider a set of microdata anonymized using a probabilistic algorithm as $k$-anonymous 
if it is $Pk$-anonymous. 

2. Its semantics is \dquote{no one estimates which person the record came from with more than $1/k$ probability (regardless of the link's actual correctness).}
From the viewpoint that 
privacy breaches are not only derived from correct information, 
this semantics is stronger than the prevention of only correct links. 

3. $Pk$-anonymity never causes failure of anonymization. 
Some current probabilistic $k$-anonymity notions are defined as \dquote{satisfaction of $k$-anonymity with certain probability.}  
Unlike these notions, 
$Pk$-anonymity always casts a definite level of re-identification hardness to the adversary
while it is defined via the theory of probability.  

4. It is non-parametric;
that is, 
no assumption on the distribution of raw microdata is necessary. 
Furthermore, 
it does not require any raw microdata to evaluate $k$. 


\paragraph{Applying $Pk$-anonymity to PRAM}
$Pk$-anonymity on PRAM is analyzed. 
The value of $k$ is derived from parameters of PRAM with no parametric assumption. 
Furthermore, 
we propose an algorithm to satisfy both $Pk$-anonymity (and $\varepsilon$-DP) with any value of $k$ (and $\varepsilon$) is given. 



\subsection{Related Work}\label{sec:related}

\subsubsection{On Probabilistic $k$-anonymity Notions}
There are many studies on $k$-anonymity, and it has many supplemental privacy notions such as 
$\ell$-diversity and $t$-closeness \cite{LLV07}.

There have also been several studies that are relevant to the probability.

Wong et al.
proposed $(\alpha, k)$-anonymity \cite{WLFW06}.
Roughly speaking, 
$(\alpha, k)$-anonymity states that (the original) $k$-anonymity is satisfied with probability $\alpha$. 
Lodha and Thomas proposed $(1-\beta, k)$-anonymity.
This is a relaxation from $k$-anonymity in a sample to that in a population.
These two notions are essentially based on the original $k$-anonymity and 
are relaxations that allow failures of anonymization in a certain probability. 
$Pk$-anonymity is fully probabilistically defined and never causes failure of anonymization. 

Aggarwal proposed a probabilistic $k$-anonymity \cite{Ag08}. 
Their goal was the same as with $Pk$-anonymity; 
however, 
it requires a parametric assumption that the distribution of raw microdata is a parallel translation of randomized microdata, 
and this seems to be rarely satisfied
since a randomized distribution is generally flatter than the prior distribution.

Soria-Comas and Domingo-Ferrer also proposed their probabilistic $k$-anonymity \cite{SD12}. 
They applied it to random swapping and micro-aggregation. 
The semantics of their anonymity is 
\dquote{no one can {\em correctly link} a person to a record with a higher probability than $1/k$}
and $Pk$-anonymity is stronger. 
Unfortunately, 
further comparison is difficult since we could not find a sufficiently formal version of the definition.

\subsubsection{On Privacy Measures Applicable to PRAM}

Aggarwal and Agrawal proposed 
a privacy measure based on conditional differential entropy \cite{AA01}. 
This measure requires both raw and randomized data to be evaluated, 
unlike $Pk$-anonymity.

Agrawal et al. proposed $(s, \rho_{1},\rho_{2})$ Privacy Breach \cite{AST05},
which is based on probability and applicable to retention-replacement perturbation.
In contrast to $k$-anonymity, it does not take into account background knowledge concerning raw data, 
that is, concerning quasi-identifier attributes.

 

Rebollo-Monedero et al.  \cite{RFD10} proposed a $t$-closeness-like privacy criterion 
and a distortion criterion which are applicable to randomization, 
and showed that PRAM can meet these criteria.
Their work was aimed at clarifying the privacy-distortion trade-off problem via information theory,
in the area of attribute estimation.
Therefore, 
they did not mention whether PRAM can satisfy a well known privacy notion
such as $k$-anonymity.


\subsubsection{On Microdata Release Algorithms Satisfying $k$-anonymity and DP}

Li et al. proposed a method satisfying $k$-anonymity and $(\varepsilon, \delta)$-DPS by combining random sampling and $k$-anonymization \cite{LQS12}. 
Since $(\varepsilon, \delta)$-DPS is based on $(\varepsilon, \delta)$-DP, PRAM's $\varepsilon$-DP is stronger. 

Soria-Comas and Domingo-Ferrer proposed methods for $t$-closeness and $\varepsilon$-DP. 
However, a certain amount of the adversary's knowledge is assumed. 
Additionally, it cannot be applied when the adversary has any knowledge about all attributes.
On the other hand, PRAM guarantees $\varepsilon$-DP regardless of the adversary's knowledge. 

\subsubsection{On Probabilistic Anonymization Algorithms Related to PRAM}
There have been several studies \cite{MS06,RH02,EGS03,AHP09} on local perturbation in which 
individuals anonymize their respective data before transferring it to some central server.

Agrawal et al. proposed a FRAPP \cite{AHP09}.
They use a specific transition probability matrix called 
MASK \cite{RH02} and Cut and paste \cite{ESAG04}
to satisfy $\rho_1$-to-$\rho_2$ privacy breach \cite{EGS03}.
After that, Rastogi et al. \cite{RHS07} proposed the $\alpha\beta$-algorithm
that improves utility.
These methods are closely related to PRAM, but
they do not consider whether PRAM can satisfy a well-known notion such as $k$-anonymity.


\subsection{Organization of Paper}

In \refsec{preliminary}, 
we discuss the notations used in the paper and preliminary definitions. 
In \refsec{pk-anonymity}, 
we propose our probabilistic $k$-anonymity, $Pk$-anonymity.
In \refsec{apply}, 
we apply $Pk$-anonymity to PRAM
and give algorithms 
for PRAM to satisfy both $\varepsilon$-DP and $Pk$-anonymity.
In \refsec{ex}, 
we describe the experimental results regarding the utility of PRAM with parameters derived from the algorithms given in the previous sections. 
Finally, we state the conclusions of this paper in \refsec{conclusions}.

\section{Preliminaries}\label{sec:preliminary}

\subsection{Basic Settings}
We consider two scenarios 
of microdata release using randomization. 
One is the setting in which a database administrator randomizes microdata (\reffig{setting}(a)). 
The other is that in which individuals randomize their own records (\reffig{setting}(b)). 
The latter is better with respect to privacy. 
PRAM is not only applicable to the former 
but also applicable to the latter \cite{AST05}
in contrast, $k$-anonymity can only be applied to the former. 
Thus, our $Pk$-anonymity is applicable to both scenarios via PRAM.

Since a person randomizes his/her data in the latter scenario, 
no one has all the raw microdata. 
Therefore, 
a person should be able to conduct appropriate randomization without another person's record. 
Fortunately, we can show that PRAM's parameter satisfying $Pk$-anonymity and DP can be determined using
only the expected record count and metadata of attributes, as mentioned in \refsec{apply}.

\begin{figure}
\begin{center}
\includegraphics[width = 7.5cm]{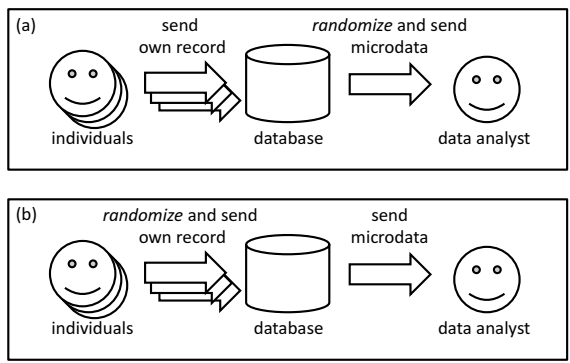}	
\caption{Two Scenarios of Microdata Release using PRAM} \label{fig:setting}
\end{center}
\end{figure}



\subsection{Notation}\label{sec:not}

We treat a table-formed database as both private and released data. 
Since  
the record count is revealed at the same time that the data are released in an ordinary microdata release,
we assume that the record count is public and static in theory.  
Furthermore, 
we consider attributes as one bundled direct product attribute
since it is sufficient for theoretical discussion.  

Basically, 
we use the following notations. 
\begin{itemize}
\item ${\cal T}$: the set of any private tables
\item $\tau, T$: a private table as an instance/random variable
\item ${\cal T}'$: the set of any released tables
\item $\tau', T'$: a released table as an instance/random variable
\item ${\cal R}, {\cal R}'$: the sets of all records in a private/released table
\item ${\cal V}, {\cal V}'$: the sets of any values in a private/released table
\item $A$: a transition probability matrix in PRAM
\item $f_{X}$: the probability function of $X$ where $X$ is a random variable
\end{itemize}
In the discussion of multi-attributes, 
we also use the following notations. 
\begin{itemize}
\item ${\cal A}, {\cal A}'$: the sets of attributes in a private/released table
\item ${\cal V}_{a}, {\cal V}'_{a'}$ where $a\in{\cal A}$ and $a'\in{\cal A}'$: 
the sets of values in a private/released table, 
i.e., ${\cal V} = \prod_{a \in {\cal A} } {\cal V}_{a}$ and ${\cal V} = \prod_{a' \in {\cal A}' } {\cal V}'_{a}$ hold where $\prod$ means the direct product.

\item $A_{a}$ where $a\in{\cal A}$: transition probability matrix of each attribute
\end{itemize}
We consider a table $\tau\in{\cal T}$ (or, $\tau'\in{\cal T}'$) as a map from ${\cal R}$ to ${\cal V}$ (or, ${\cal R}'$ to ${\cal V}'$). More formally, we difine $\tau$ (or, $\tau'$) as follows.

\begin{definition}(tables)\\
Let a record set ${\cal R}$ and a value set ${\cal V}$ be finite sets. 
Then, the following map $\tau$ is called a table on $({\cal R}, {\cal V})$.  
\begin{equation*}
\tau   : {\cal R}    \to    {\cal V},
\end{equation*}
When we discuss a multi-attribute table, 
${\cal V}$ is represented as $\prod_{a \in    {\cal A}   }{{\cal V}   _{a}}$, 
where an attribute set ${\cal A}$ is a finite set, 
each ${\cal V}_{a}$ is also a finite set for any $a\in{\cal A}$,
and $\prod$ means the direct product.
\end{definition}

\subsection{PRAM}

PRAM \cite{KWG97} was proposed by Kooiman et al. in 1997
as a privacy preserving method for microdata release. 
It changes data according to a {\em transition probability matrix}. 
A transition probability matrix consists of probabilities in which each value in a private table will be changed into other specific (or the same) values. 
$A_{u,v}$ denotes the probability $u \in \mathcal{V}$ is changed into $v\in \mathcal{V}'$.
For example, $A_{\text{male},\text{female}}$ means \dquote{male $\to$ female} is $25\%$.

PRAM is a quite general method. 
Invariant PRAM \cite{KWG97}, 
retention-replacement perturbation \cite{AST05}, 
etc., 
are known as instantiations of it. 
Specifically, 
retention-replacement perturbation is simple and convenient. 

\subsubsection{Retention-replacement Perturbation}\label{sec:rrp}

In retention-replacement perturbation, 
individuals probabilistically replace their data with random data
using given {\it retention probability} $\rho$. 
First, data are retained with $\rho$, 
and if the data are not retained, 
they will be replaced with a uniformly random value chosen from the attribute domain.
Note that even if data are not retained, 
there is still the possibility that the data will not be changed, 
because the data value is included in the attribute domain as well as other values.
For example, for an attribute \dquote{sex,}
when $\rho = 0.5$, \dquote{male} is retained with $1/2$ probability, 
and with the remaining $1/2$ probability, 
it is replaced with a uniformly random value, 
namely, a value \dquote{female} and a value \dquote{male,} which is the same as the original, both with $1/2 \times   1/2 = 1/4$ probability.
Eventually, the probability that \dquote{male} changes into \dquote{female} is $1/4$, 
and the probability that it does not change is $3/4$. 
The lower the retention probability, 
the higher privacy is preserved.
On the contrary, the lower the probability, 
the lower utility.
These probabilities form the following transition probability matrix. 
\begin{eqnarray*}
\left[
\begin{array}{cc}
0.75 & 0.25 \\
0.25 & 0.75\\
\end{array}
\right]
\end{eqnarray*} 
Generally, the transition probability matrix $A_{a}$ of an attribute $a$ is written as
\begin{equation*}
(A_{a})_{v_{a}, v'_{a}} = \begin{cases}
\rho   _{a} + \dfrac{(1-\rho   _{a})}{|{\cal V}   _{a}|} & \text{if }v_{a} = v'_{a} \\
\dfrac{(1-\rho   _{a})}{|{\cal V}   _{a}|} & \text{otherwise}
\end{cases}
\end{equation*}
where for any $v \in    {\cal V}  $ and $a \in    {\cal A}  $, $v_{a} \in  {\cal V}  _{a}$ is an element of $v \in  {\cal V}  $ corresponding to $a$, 
and $\rho   _{a}$ is the retention probability corresponding to $a$.


\subsection{{\it k}-anonymity}\label{sec:k-anonymity}

The $k$-anonymity \cite{SS98} \cite{Sw02} is a privacy notion that is applicable to table-formed databases and defined as \dquote{for all database records, there are at least $k$ records whose values are the same,} 
in other words, \dquote{no one can narrow down a person's record to less than $k$ records.}

Using the notations in \refsec{not}, 
we represent the definition of $k$-anonymity \cite{Sw02} as follows. 
\begin{definition}{\em(}{\it $k$-anonymity}{\em)}

For a positive integer $k$, a released table $\tau'\in{\cal T}'$ is said to satisfy
 $k$-anonymity (or to be $k$-anonymous), 
if and only if it satisfies the following condition. 

\begin{quote}
For any $r' \in {\cal R}'$, there are $k$ or more $\hat{r'}$'s such that 
$\hat{r'}\in {\cal R}'$ and $\tau'(r') = \tau'(\hat{r'})$.
\end{quote}
\end{definition}
A released table $\tau'$ in the above definition represents all columns corresponding to quasi-identifier attributes of an anonymized table. 

However, the definition in \cite{Sw02} is problematic; i.e., 
there are some tables that satisfy $k$-anonymity but do not achieve its aim. 
For example, 
a table generated by copying all a private table's records $k$ times satisfies $k$-anonymity but 
it is obviously not safe.
Therefore, we assume $|{\cal R}| = |{\cal R}'|$ to strengthen the above definition in the discussion of $k$-anonymity in this paper.

\subsection{Anonymization and Privacy Mechanisms}
We define anonymization and privacy mechanisms separately
to discuss them formally.
First we define anonymization.
\begin{definition}\label{def_ano}
(anonymization)\\
Let ${\cal R}$, ${\cal R}'$, ${\cal V}$ and ${\cal V}'$ be finite sets,  
${\cal T}$ and ${\cal T}'$ be the sets of all tables on $({\cal R}, {\cal V})$ and $({\cal R}', {\cal V}')$, respectively, 
and let $\pi$ be a map $\pi: {\cal R} \to {\cal R}'$. 
Then, for any $\tau  \in  {\cal T}$ and $\tau'  \in  {\cal T}  '$, 
a map $\delta   : {\cal T}   \to    ({\cal R}    \to    {\cal V}   ')$ is called anonymization with $\pi$ from $\tau$ to $\tau'$ if and only if they satisfy
\begin{equation}
\delta   (\tau   ) = \tau   ' \circ    \pi  ,  \label{eq:delta_pi}
\end{equation}
where the notation ${\cal X}    \to    {\cal Y} $ denotes the set of all maps from ${\cal X}$ to ${\cal Y}$ for any set ${\cal X}$ and ${\cal Y}$.
\end{definition}

Anonymization $\delta$ represents an anonymization algorithm 
such as perturbation, $k$-anonymization, etc.
A map $\pi$ represents an anonymous communication channel, the shuffling function, or another component
which hides the order of records in $\tau$.
In this paper, we adopt the uniformly random permutation as $\pi$.\footnote{A map $\pi$
is essential for anonymization. For example, if the first record in the private table
is to be the first record in the released table, identification is trivial.}

Privacy mechanisms involve not only $\delta$ but also $\pi$, ${\cal R}$, ${\cal R}'$, ${\cal V}$ and ${\cal V}'$, 
and random variables are brought to extend the above definitions to probabilistic ones. 
Random variables corresponding to $\tau  $, $\tau  '$, $\pi  $, and $\delta  $ are denoted by $T$, $T'$, $\Pi  $, and $\Delta  $, 
respectively.  
We assume $T$, $\Pi  $, and $\Delta  $ are mutually independent as probabilistic events,  
while $T'$ is dependent on the other three random variables. 

\begin{definition}(privacy mechanisms)\\
Let ${\cal R}$, ${\cal R}'$, ${\cal V}$, ${\cal V}'$, ${\cal T}$, and ${\cal T}'$ be the same as Definition~\ref{def_ano}, 
and let $T$, $T'$, $\Pi$, and $\Delta$ be random variables on ${\cal T}$, ${\cal T}'$, ${\cal R} \to {\cal R}'$, and ${\cal T} \to ({\cal R} \to {\cal V}')$, respectively, 
such that $T$, $\Pi$, and $\Delta$ are mutually independent as probabilistic events,  
where the notation ${\cal X}    \to    {\cal Y} $ denotes the set of all maps from ${\cal X}$ to ${\cal Y}$ for any set ${\cal X}$ and ${\cal Y}$.
Then, the $6$-tuple $({\cal R}   , {\cal V}   , {\cal R}   ', {\cal V}   ', \Pi   , \Delta   )$ is called a {\it privacy mechanism} from $T$ to $T'$ if and only if they satisfy the following equation. 
\begin{equation*}
\Delta   (T   ) = T   ' \circ    \Pi    \label{eq:Delta_Pi}
\end{equation*}

\end{definition}

\subsection{Differential Privacy on PRAM}
Dwork proposed DP \cite{Dw06} in 2006. 
It results in 
\dquote{an (statistical) output not changing much even if a database is changed with respect to at most one person.}
Since it can be satisfied regardless of adversaries, 
it is being widely studied. 

Differential privacy is defined with a real number parameter $\varepsilon$. 
\begin{definition}{\em(}{\it $\varepsilon$-DP}{\em)}\\
Let $\mathscr{D}$ be a set of databases and $d$ be a non-negative integer. 
A privacy mechanism ${\cal K} : \mathscr{D} \to \mathbb{R}^{d}$ is a probabilistic algorithm, 
and $\varepsilon$ is a {\rm (}small{\rm )} positive real number. 
We say ${\cal K}$ gives $\varepsilon$-DP if, and only if 
for $S \subseteq {\rm Range}({\cal K})$ and 
any pair $D_{1}, D_{2}$ of databases \dquote{differing at most by $1$ element,} 
the following condition is satisfied.  
\vspace{2ex}
\begin{equation}\label{fm:1}
\Pr[{\cal K}(D_{1})\in S]\leq\exp(\varepsilon)\Pr[{\cal K}(D_{2})\in S]
\end{equation}
\end{definition}
Note that what 'databases' and \dquote{differing at most by $1$ element} mean remains free to interpretation.

Differential privacy is used as a privacy notion on interactive statistical databases as usual. 
However, PRAM is known to satisfy $\varepsilon$-DP for the query
\dquote{\texttt{select * from $\tau$}} in SQL manner \cite{LWR12}.
The query obviously represents the release of microdata.
We introduce the known result \cite{LWR12}
and discuss $\varepsilon$-DP on PRAM in addition to $Pk$-anonymity. 

PRAM satisfies $\varepsilon$-DP with the following parameters \cite{LWR12}. 
\begin{theorem}\label{th:dp2}
For any PRAM mechanism $\Delta$ whose transition probability matrix is denoted by $A$,  
$\Delta$ gives $\varepsilon$-DP with the following $\varepsilon$. 
$$\varepsilon = \ln \dmax{u, v \in {\cal V} \atop v' \in {\cal V}'}\dfrac{A_{u, v'}}{A_{v, v'}}$$ 
\end{theorem}

The theorem has been already shown but it may not be rigorous and 
suit our notations. Therefore, we give another proof of \refth{dp2} in \refapp{dp}.
We show a multi-attribute representation of \refth{dp2} below.  
Simply, $\varepsilon$ becomes the summation of each attribute's $\varepsilon$. 
\begin{corollary}
For any PRAM mechanism $\Delta$ whose transition probability matrices are $A_{a}$ for each attribute $a\in{\cal A}$,  
$\Delta$ gives $\varepsilon$-DP with the following $\varepsilon$. 
$$\varepsilon = \dsum{a\in{\cal A}}\ln \dmax{u, v \in {\cal V}_{a} \atop v' \in {\cal V}'_{a}}\dfrac{(A_{a})_{u, v'}}{(A_{a})_{v, v'}}$$
\end{corollary}

Regarding retention-replacement perturbation, 
$\varepsilon$ is evaluated as follows. 
\begin{corollary}\label{cor:dp3}
For any retention-replacement perturbation $\Delta$ whose retention probabilities of each attribute are $\rho_{a}$, 
$\Delta$ gives $\varepsilon$-DP privacy with the following $\varepsilon$. 
\begin{equation}\label{eq:rrp_dp}
\varepsilon = \dsum{a\in{\cal A}}\ln\dfrac{1 + (|{\cal V}_{a}| - 1)\rho_{a}}{1 - \rho_{a}}
\end{equation}
\end{corollary}


\section{{\it P\kkk}-anonymity}\label{sec:pk-anonymity}

As the name suggests, 
$k$-anonymity represents anonymity among privacy notions. 
It is known that satisfying only anonymity is not enough to preserve privacy \cite{MGK06};
thus, 
further privacy notions that prevent attribute estimation were developed after $k$-anonymity. 
However, 
this \textit{never} means that \dquote{anonymity is unnecessary.}
These stronger privacy notions rely on the assumption that $k$-anonymity has already been satisfied. 
Therefore, the same as with deterministic microdata release, 
we consider anonymity as the first privacy requirement in randomization-based microdata release.  
Regarding randomization, however, 
anonymity has not yet been clarified.  
Obviously, 
this is a critical problem and should be solved as soon as possible. 

\subsection{Problem with {\it k}-anonymity on Randomization}

We now explain what occurs when one applies $k$-anonymity directly to a randomized table.
Imagine that one randomizes all records' quasi-identifiers uniformly randomly. 
Furthermore, suppose that the resulting table happens to have a record whose data are unique. 
The randomized table does not satisfy $k$-anonymity
because it has a unique record. 
However, 
an adversary cannot identify anyone's record (without knowledge of sensitive attributes)
since uniformly random values provide no information. 
In other words, 
the table should be considered as fully anonymous, although the table does not satisfy $k$-anonymity. 
Therefore, 
we need a new definition of $k$-anonymity applicable to randomization. 

\subsection{Intuitive Requirement}

To apply $k$-anonymity to randomization, 
we have to determine what kind of notion we should construct.
Intuitively, 
\dquote{no one can choose the correct record of a person with probability $1/k$} 
is the likely choice. 
However, we have to take into account an adversary's incorrect presumption. 
Regarding privacy, 
the problem is not only the leakage of correct information, 
but the creation of \textit{incorrect} information about a person. 
Since a person does not wish to reveal correct private information, 
neither the person nor the administrator of the database can resolve the adversary's misconception. 
Therefore, 
we require a stronger sentence, 
\dquote{no one {\em estimates} which person the record came from with more than $1/k$ probability.} 
Note that this second sentence involves the first sentence (no one can choose\ldots), 
because an adversary who correctly chooses the record of a person 
with probability $1/k$ is able to estimate the record at confidence $1/k$. 

\subsection{Background Knowledge of Adversary}
In the definition of $k$-anonymity, there is no adversary, 
and this definition is described as a simple condition to be satisfied in a table. 
This is convenient for measuring $k$-anonymity.
At the same time, however,
it makes the meaning of privacy unclear. 

Therefore, there is an adversary in our model of $Pk$-anonymity.
The probability of linkage is varied according to 
the background knowledge of the adversary.
In the $Pk$-anonymity model, 
an adversary's background knowledge is represented as 
a probabilistic function $f_{T}$\footnote{It means that the
adversary knows that the private table is $\tau_{1}$ 
with probability $x_{1}$, $\tau_{2}$ with $x_{2}$ and so on.
It is not a distribution of values in a specific table, 
but the distribution on the space of all tables.} 
on the private table.
$Pk$-anonymity requires the privacy mechanism of that
the probability of linkage is bounded by $1/k$
for {\em all} $f_{T}$.
It means that we deal with an adversary 
who has arbitrary knowledge about the private table:
The adversary might know the private table itself and incorrect private tables. 

We note that even if in the extreme case where 
the adversary knows the private table itself,
$Pk$-anonymity can be satisfied by using
the randomness in the privacy mechanisms.
Of course, we assume that the adversary knows the released table,
the anonymization algorithm, and parameters used in the system
in addition to the background knowledge.

\ignore{
\subsection{Adversary}

There is an adversary in our model of $Pk$-anonymity.
In the definition of $k$-anonymity, there is no adversary, 
and this definition is described as a simple condition to be satisfied in a table. 
This is convenient for measuring $k$-anonymity.
At the same time, however,
it makes the meaning of privacy unclear. 
We define $Pk$-anonymity the same as its intuitive semantics, 
\dquote{no one(adversary) estimates which person the record came from with more than $1/k$ probability,} 
to clarify its meaning. 

In the $Pk$-anonymity model, 
an adversary is characterized by his/her background knowledge about a private table. 
The knowledge is represented as a probabilistic function $f_{T}$ on the private table
(i.e., the adversary knows that the private table is $\tau_{1}$ 
with probability $x_{1}$, $\tau_{2}$ in $x_{2}$, \ldots.
Note that it is not a distribution of values in a specific table, 
but the distribution on the space of all tables). 

In fact, we do not believe we can restrict either an adversary's knowledge or his/her power. 
He/she may possess some knowledge about the target person;
moreover, he/she might have almost all knowledge about the raw private table. 
Therefore, we deal with arbitrarily strong adversaries who might possess arbitrary knowledge about the private table, 
including the private table itself 
and incorrect private tables. 
The only thing we rely on is the randomness in randomization. 
Needless to say, 
we assume an adversary can browse the released table 
and identify the anonymization algorithm and parameter used in the system. 
}

\subsection{Definition of {\it Pk}-anonymity}

We define our new anonymity, $Pk$-anonymity 
and $Pk$-anonymization, which is a privacy mechanism that always satisfies $Pk$-anonymity. 

First, we define an attack by an adversary with background knowledge, which
is represented as an estimation by the following probability, 
where $\tau'$ is a released table, $\Pi$ is a uniformly random injective map from ${\cal R}$ to ${\cal R}'$, $r\in  {\cal R}  $, $r'\in  {\cal R}  '$ and ${\Delta}   (T   ) = T   ' \circ    \Pi$. 
\begin{equation}
\pr[\Pi   (r) = r' | T' = \tau']
\end{equation}
The term ${\cal R}$ represents a set of individuals, 
and ${\cal R}'$ represents a set of record IDs (not necessarily explicit IDs. In anonymized microdata, it maybe just a location in storage.). 
The $\Pi$'s randomness represents that \dquote{an adversary has no knowledge of the linkage between individuals and the records in $\tau'$.}
Taken together, 
the above probability represents the following probability from the standpoint of an adversary who saw $\tau'$. 
$$\pr[\text{a person }r\text{'s record in }\tau'\text{ is }r']$$
We denote the above probability as ${\cal E}  (f_{T}, \tau  ', r, r')$. 

Next, we define $Pk$-anonymity.


\begin{definition}\label{def_pk}{\em(}{\it $Pk$-anonymity}{\em)}\\
Let ${\cal R}   $, $ {\cal V}$, ${\cal R}   '$, and ${\cal V}   '$ be finite sets, 
and $\Pi$ and $\Delta$ be random variables on ${\cal R} \to {\cal R}'$ and ${\cal T} \to ({\cal R} \to {\cal V}')$, respectively, 
where ${\cal T}$ denotes the set of tables on $({\cal R}   ,  {\cal V})$ and 
the notation ${\cal X}    \to    {\cal Y} $ denotes the set of all maps from ${\cal X}$ to ${\cal Y}$ for any set ${\cal X}$ and ${\cal Y}$. 
Furthermore, let $\tilde{\Delta}$ denote a 6-tuple $({\cal R}   , {\cal V}   , {\cal R}   ', {\cal V}   ', \Pi   , \Delta   )$. 

Then, for any real number $k \geq   1$ and a table $\tau'$ on $({\cal R}   ', {\cal V}   ')$, 
a pair $(\tilde{\Delta}, \tau   ')$ is said to satisfy {\it $Pk$-anonymity} (or to be $Pk$-anonymous) if and only if for any random variables $T$ of tables on $({\cal R}   , {\cal V})$ and $T'$ of tables on $({\cal R}   ', {\cal V}   ')$ such that $\tilde{\Delta}$ is a privacy mechanism from $T$ to $T'$, 
any record $r \in  {\cal R}  $ of the private table $T$ and any record $r' \in   {\cal R}  '$ of the released table $\tau'$, the following equation is satisfied. 
$$\pr[\Pi   (r) = r' | T' = \tau'] \leq \dfrac{1}{k}$$
\end{definition}

\begin{definition}{\em(}{\it $Pk$-anonymization algorithms}{\em)}\\
Let ${\cal R}   $, $ {\cal V}$, ${\cal R}   '$, ${\cal V}   '$, $\Pi$, $\Delta$, and $\tilde{\Delta}$ be the same as Definition~\ref{def_pk}, 
and let ${\cal T}'$ denote the set of all tables on $({\cal R}', {\cal V}')$. 

Then, for any real number $k \geq   1$, 
$\tilde{\Delta}$ is said to be a $Pk$-anonymization if and only if $(\tilde{\Delta}  , \tau  ')$ satisfies $Pk$-anonymity for any released table $\tau  '\in{\cal T}'$ such that there exists a private table $\tau$ on $({\cal R}   ,  {\cal V})$ which satisfies $\pr[\Delta(\tau) = \tau'\circ \Pi] \neq 0$. 
\end{definition}
we treat only $\Delta$ within 6-tuple of a privacy mechanism $({\cal R}   , {\cal V}   , {\cal R}   ', {\cal V}   ', \Pi   , \Delta   ) = \tilde{\Delta}$;
thus, 
we do not differentiate $\tilde{\Delta}$ and $\Delta$.

$Pk$-anonymity's direct meaning is 
\dquote{no one estimates which person the record came from with more than $1/k$ probability.} 
Intuitively, 
it seems to be similar to \dquote{no one can narrow down a person's record to less than $k$ records,} 
which is an intuitive concept of $k$-anonymity. 
This intuitive similarity can also be confirmed mathematically. 
Furthermore, 
as far as deterministic anonymization algorithms, such as $k$-anonymization algorithms, are concerned, 
two anonymity notions can be shown to be equivalent to each other.
Therefore, we say $k$-anonymity is satisfied in a randomized table if $Pk$-anonymity is satisfied in the table. 

\begin{theorem}\label{th:eq} 
For any positive integer $k$, privacy mechanism ${\Delta}$, and released table $\tau'$, the following relation holds
if ${\Delta}$ is deterministic, i.e., 
for any $\tau\in{\cal T}$, there exists unique anonymized table $\hat{\tau}$ and ${\Delta}(\tau) = \hat{\tau}$.
\begin{quote}
$\tau  '$ is $k$-anonymous $\Leftrightarrow  $ $({\Delta}  , \tau  ')$ is $Pk$-anonymous
\end{quote}
\end{theorem}
This theorem represents equality of $Pk$-anonymity and $k$-anonymity under the consideration of deterministic anonymization algorithms, 
which are the applicable field of $k$-anonymity.  
Therefore, $Pk$-anonymity is deemed as an extension of $k$-anonymity. 

\medskip
\noindent
{\it (Proof of \refth{eq})}\\
This theorem is shown with the following two lemmas.

\begin{lemma}\label{lm:ktoPk}
For any positive integer $k$,  
if a released table $\tau  '$ is $k$-anonymous, then
$({\Delta} , \tau  ')$ is $Pk$-anonymous for any privacy mechanism ${\Delta}$. 
\end{lemma}

\begin{lemma}\label{lm:Pktok}
For any real number $t \geq  1$, positive number $k$ such that $k \leq   t$,  
any deterministic privacy mechanism ${\Delta} $, and released table $\tau  '$, 
if $({\Delta}   , \tau  ')$ is $Pt$-anonymous, then $\tau  '$ is $k$-anonymous. \\
\end{lemma}

Roughly, 
\reflm{ktoPk} states that \dquote{$k\Rightarrow Pk$ always,}
and \reflm{Pktok} states that \dquote{$Pk\Rightarrow k$ if an anonymization algorithm is deterministic.}


\medskip
\noindent
({\it Proof of \reflm{ktoPk}})\label{ap:lm1}\\
First, we use notation $\sharp   _{\tau   '}(v')$ as $|\tau   '^{-1}(\{   v'\}   )|$, 
and say $r'$ is $k$-anonymous in $\tau  '$ if $\sharp   _{\tau   '}(\tau   '(r')) \geq   k$. 
Then, 
$k$-anonymity of $\tau'$ is represented as
\dquote{$r'$ is $k$-anonymous in $\tau  '$ for any $r'\in   {\cal R}   '$.}

As mentioned in \refsec{pk-anonymity}, 
we show \reflm{ktoPk} and \reflm{Pktok}. 
Note that the following equality holds by definition. 
\begin{equation}\label{eq:delta_pi}
\Delta   (T   ) = T   ' \circ    \Pi
\end{equation}

We show that an estimation probability, ${\cal E} (f_{T}, \tau', r, r')$,
is equal to or less than $1/k$.
For any background knowledge $f_{T} : {\cal T}  \to  \mathbb{R}$, 
any $r \in {\cal R}$ and any $r' \in {\cal R}'$, the following equations hold.\\

\dhead{{\cal E}   (f_{T}, \tau   ', r, r') = \pr[\Pi   (r) = r' | T' = \tau   ']}

\dcontx{ = \dfrac{\pr[\Pi   (r) = r' \land    T' = \tau   ']}{\pr[T' = \tau   ']}
 = \dfrac{\pr[\Pi   (r) = r' \land    \Delta   (T) = \tau   ' \circ    \Pi   ]}{\pr[\Delta   (T) = \tau   ' \circ    \Pi   ]}
}
{from \refeq{delta_pi}}

\dcontx{= \dfrac{ \dsum   { \delta    : {\cal T} \to    ({\cal R}    \to    {\cal V}   ') \atop \tau\in{\cal T}   }{\!\!\!\!\!\! f_{\Delta   }(\delta   ) f_{T}(\tau   ) \pr[\Pi   (r) = r' \land    \delta   (\tau   ) = \tau   ' \circ    \Pi   ]}}
{\dsum   {\delta    : {\cal T} \to    ({\cal R}    \to    {\cal V}   ') \atop \tau    \in{\cal T}   }{\!\!\!\!\!\! f_{\Delta   }(\delta   ) f_{T}(\tau   ) \pr[\delta   (\tau   ) = \tau   ' \circ    \Pi   ]}}}{since $T, \Delta$, and $\Pi  $ are independent of each other}

We define two propositions $\Phi   (\delta   , \tau   )$ and $\hat{\Phi   }(\delta   , \tau   )$ as
\begin{equation*}
\Phi   (\delta   , \tau   ) = [\text{There exists } \hat{\pi   } : {\cal R}    \to    {\cal R}   ' \text{ such that } \delta   (\tau   ) = \tau   ' \circ    \hat{\pi  }]
\end{equation*}
\begin{equation*}
\hat{\Phi   }(\delta   , \tau   ) = [\Phi   (\delta   , \tau   ) \text{ and } (\delta   (\tau   ))(r) = \tau   '(r')]
\end{equation*}
respectively.  
Since $\Pi$ is a uniformly random permutation,
the following equations hold. \\

\dhead{\pr[\delta   (\tau   ) = \tau   ' \circ    \Pi   ]
 = 
\begin{cases}
\dfrac{
	\dprod{v' \in \im(\tau')}
	{\sharp_{\tau'}(v')!}}{|{\cal R}|!} & (\text{if } \Phi \text{ holds})\\
0 & (\text{otherwise})
\end{cases}}

\dhead{\pr[\Pi   (r) = r' \land    \delta   (\tau   ) = \tau   ' \circ    \Pi   ]}
 
\dcont{=
\begin{cases}
\dfrac{
	\left(
		\sharp_{\tau'}(\tau'(r')) - 1
	\right)!
	\!\!\!\!\!\!\!\!
			\dprod{v' \in \im(\tau') \setminus \{\tau'(r')\}}
		\!\!\!\!\!\!\!\!\sharp_{\tau'}(v')!
}
	{ |{\cal R}|!}
		 & (\text{if }\hat{\Phi}\text{ holds})\\
0 & (\text{otherwise})
\end{cases}
}
\dcont{=
\begin{cases}
\dfrac{
	\dprod{v' \in \im(\tau')}
		\sharp_{\tau'}(v')!}
	{\sharp_{\tau'}(\tau'(r')) |{\cal R}|!}
		 & (\text{if }\hat{\Phi}\text{ holds})\\
0 & (\text{otherwise})
\end{cases}
}

Therefore, the primary equation $\pr[\Pi   (r) = r' | T' = \tau   ')$ is transformed as 

\dhead{\dfrac{\dsum   { \hat{\Phi}(\delta   , \tau   )}{f_{\Delta   }(\delta   ) f_{T}(\tau   ) \dfrac{\dprod   {v' \in    \im(\tau   ')}{\sharp   _{\tau   '}(\tau   '(s'))!}}{\sharp   _{\tau   '}(\tau   '(r'))|{\cal R}   |!}}}
{\dsum   {\Phi   (\delta   , \tau   )}{f_{\Delta   }(\delta   ) f_{T}(\tau   ) \dfrac{\dprod   {v' \in    \im(\tau   ')}{\sharp   _{\tau   '}(\tau   '(s'))!}}{|{\cal R}   |!}}}
\ \ \leq    \ \ 
\dfrac{\dsum   { \Phi   (\delta   , \tau   ))}{f_{\Delta   }(\delta   ) f_{T}(\tau   ) \dfrac{\dprod   {v' \in    \im(\tau   ')}{\sharp   _{\tau   '}(\tau   '(s'))!}}{\sharp   _{\tau   '}(\tau   '(r'))|{\cal R}   |!}}}
{\dsum   {\Phi   (\delta   , \tau   )}{f_{\Delta   }(\delta   ) f_{T}(\tau   ) \dfrac{\dprod   {v' \in    \im(\tau   ')}{\sharp   _{\tau   '}(\tau   '(s'))!}}{|{\cal R}   |!}}}
}
\fright{(since $\hat{\Phi} \Rightarrow \Phi$)}

\dcont{ = \dfrac{1}{\sharp   _{\tau   '}(\tau   '(r'))} \leq    \frac{1}{k}.}
\fright{(from $k$-anonymity)}

\hfill$\Box$(\reflm{ktoPk})

\medskip
\noindent
({\it Proof of \reflm{Pktok}})\label{ap:lm2}\\
In the proof we use and show the following contraposition. 

\begin{quote}{\em
For any privacy mechanism $\Delta$, if $\tau'$ is not $k$-anonymous, 
then $(\Delta, \tau')$ is also not $Pt$-anonymous.
}\end{quote}

We consider the background knowledge, $f_{T}$, satisfying $f_{T}(\tau) = 1$. 
Let $r'\in {\cal R}'$ be a record that is not $k$-anonymous in $\tau  '$
and that satisfies $r\in   \pi   ^{-1}(r')$. 

As in the proof of \reflm{ktoPk}, the following equation holds.

\[
{\cal E}   (f_{T}, \tau   ', r, r')
	= \dfrac{
		\dsum   {\hat{\Phi   }(\delta   , \tau   )}{f_{\Delta   }(\delta   ) f_{T}(\tau   )}
			\dfrac{\dprod   {v' \in    \im(\tau   ')}{\sharp   _{\tau   '}(v')!}}
												{\sharp   _{\tau   '}(\tau   '(r'))|{\cal R}   |!}
	}{
		\dsum   {\Phi   (\delta   , \tau   )}{f_{\Delta   }(\delta   ) f_{T}(\tau   )} 
			\dfrac{\dprod   {v' \in  \im(\tau   ')}{\sharp_{\tau   '}(v')!}}{|{\cal R} |!}}
\]
Since $\Delta  $ is deterministic and $f_{T}(\tau   ) = 1$,
we transform the above equation as follows.

\[
	\dfrac{
		\dsum   {\hat{\Phi   }(\delta   , \tau   )}{f_{\Delta   }(\delta   ) f_{T}(\tau   )}
			\dfrac{\dprod   {v' \in    \im(\tau   ')}{\sharp   _{\tau   '}(v')!}}
												{\sharp   _{\tau   '}(\tau   '(r'))|{\cal R}   |!}
	}{
		\dsum   {\Phi   (\delta   , \tau   )}{f_{\Delta   }(\delta   ) f_{T}(\tau   )} 
			\dfrac{\dprod   {v' \in    \im(\tau   ')}{\sharp   _{\tau   '}(v')!}}{|{\cal R}   |!}
	}
	=
 \dfrac{
		\dfrac{\dprod   {v' \in    \im(\tau   ')}{\sharp   _{\tau   '}(v')!}}
												{\sharp   _{\tau   '}(\tau   '(r'))|{\cal R}   |!}
	}{
		\dfrac{\dprod   {v' \in    \im(\tau   ')}{\sharp   _{\tau   '}(v')!}}{|{\cal R}   |!}
	}
=
\dfrac{1}{\sharp   _{\tau   '}(\tau   '(r'))}.
\]
We assume $r'$ is not $k$-anonymous; 
therefore, $\dfrac{1}{\sharp   _{\tau   '}(\tau   '(r'))} \gneq \dfrac{1}{k}$.
\fright{ 
$\Box  $(\reflm{Pktok})}

The above two lemmata immediately imply \refth{eq}.
\begin{flushright}
$\Box   $
\end{flushright}

Furthermore, $k$-anonymization and $Pk$-anonymization also have a similar equality.  

\begin{corollary}
For any positive integer $k$ and privacy mechanism ${\Delta}$, 
if ${\Delta} $ is deterministic, the following holds.

\begin{quote}
${\Delta}  $ is $k$-anonymization $\Leftrightarrow  $ ${\Delta} $ is $Pk$-anonymization
\end{quote}
\end{corollary}


Through \refth{eq},
we have seen that $Pk$-anonymity is an exact mathematical extension of $k$-anonymity.
Moreover, the intuitive meaning of $k$-anonymity, 
\dquote{no one can narrow down a person's record to less than $k$ records} 
is applicable from the following viewpoint.
Under a privacy mechanism ${\Delta}$ and a certain released table $\tau'$, 
an adversary's estimation ${\cal E}(f_{T}, \tau', r, r')$ is $1/k$ or less for any $r\in{\cal R}$ and $r'\in{\cal R}'$, 
when $({\Delta}, k)$ is $Pk$-anonymous.
Then by definition, for any $k - 1$ records $\{ r'_{i}\}_{0 \leq i < k - 1}$ in $\tau'$, the following relation holds. 
$$\dsum{0 \leq i < k - 1}{\cal E}(f_{T}, \tau', r, r'_{i}) \leq 1 - \dfrac{1}{k} \lneq 1$$
This relation means that when one has chosen $k - 1$ records from $\tau'$, 
there is always $1/k$ probability that $r$ is not in these $k - 1$ records in $\tau'$.
This precisely means that \dquote{no one can narrow down a person's record to less than $k$ records.}\\

Remember that an adversary is considered as background knowledge and a distribution.
In the field of cryptography, an adversary is often represented as an algorithm.  
We show that an adversary represented as a probabilistic algorithm $M$ that takes inputs as $(\tau  ', r)$ cannot select $r$'s record in a released table with a higher probability than $1/k$.

\begin{proposition}
For any $Pk$-anonymization ${\Delta}$, $\tau\in T$, $\tau'\in T'$ such that ${\Delta}(T) = T'\circ\Pi$, $r\in {\cal R}$ and probabilistic algorithm $M$ that takes $\tau$ and $r$ as inputs, 
$M$ do not select $r'\in{\cal R'}$ such that $\Pi(r) = r'$ with a higher probability than $1/k$.
\end{proposition}

\medskip
\noindent
{\it (Proof of Proposition 1)\\}
Let $f_{T}$ be the following probability function. 
$$\pr[T = \tau] = \begin{cases}
1 & \text{ if }\tau = \tau_{t}\\
0 & \text{ otherwise}
\end{cases}$$
Under this $f_{T}$, 
the probability $\pr[\Pi(r) = r'|T' = \tau']$ is not only an adversary's estimate, 
but also the true probability. 
On the other hand, it is $1/k$ or smaller by $Pk$-anonymity;
therefore, 
no function selects $r'$ with a higher probability than $1/k$, 
and $M$ is only a function.
\\\hfill$\Box$

\section{Applying {\it Pk}-anonymity (and DP) to PRAM}\label{sec:application}\label{sec:apply}

We apply $Pk$-anonymity to PRAM. 
First, we show a theorem on general PRAM for calculating $k$. 
Next, we describe a more concrete formula on retention-replacement perturbation introduced in \refsec{rrp}. 
Finally, combining existing result,
we propose an algorithm to satisfy both $Pk$-anonymity and $\varepsilon$-DP.

We assume ${\cal V}    = {\cal V}   '$.
The privacy mechanism ${\Delta}$ is defined along with PRAM, 
i.e., defined for any $r \in   {\cal R}   $ and $v' \in    {\cal V}   '$, as follows.
\begin{equation*}
f_{({\Delta}   (T))(r)}(v') = {A_{T(r), v'}}
\end{equation*}
We call such a privacy mechanism a PRAM mechanism.

\begin{theorem}$($ $Pk$-anonymity on PRAM$)$\\\label{th:2}
A PRAM mechanism whose transition probability matrix is $A$ is a $Pk$-anonymization if and only if $k$ is described as follows.
\begin{equation*}
k \leq 1 + (|{\cal R}| - 1) \dmin{u, v\in{\cal V} \atop u', v' \in{\cal V}'}\dfrac{A_{u, v'}A_{v, u'}}{A_{u, u'}A_{v, v'}}
\end{equation*}
\end{theorem}
Note that this theorem shows the tight bound of $k$.
This theorem is shown by evaluating the maximum probability of estimation ${\cal E}(f_{T}, \tau', r, r')$ on  
$r\in{\cal R}, r'\in{\cal R}', \tau'\in{\cal T}'$ and background knowledge $f_{T}:{\cal T}\to\mathbb{R}$. 
The probability takes the maximum value in the following case. 
\begin{itemize}
\item All values in private table $\tau$ happened to be retained in released table $\tau'$
\item There are only two values in $\tau$ and $\tau'$, one is $\tau(r)$ and the other is $v\in{\cal V}$, 
which satisfies $\tau(s) = v$ for any record $s\neq r$ 
\item $\tau(r)$ and $v$ shown above are different from each other in all attributes
\item The adversary knows all about the private table, i.e.,
$f_{T}(\tau) = 
\begin{cases}
1 & \text{if } \tau = \tau' \\
0 & \text{otherwise}
\end{cases}$
\end{itemize}
With this fact, $k$ can be derived by substituting each parameter in estimate ${\cal E}(f_{T}, \tau', r, r')$.

\medskip
\noindent
({\it Proof of \refth{2}})\label{ap:th2}

We show this theorem by evaluating the maximum of ${\cal E}   (f_{T}, \tau   ', r, r')$ on  
$r\in  {\cal R}   , r'\in  {\cal R}   ', \tau  '\in{\cal T}'$ and $f_{T}:{\cal T} \to    \mathbb{R}$. 
Similar to the proof of \reflm{ktoPk}, the following equation holds.\\

\dhead{{\cal E}   (f_{T}, \tau   ', r, r') = \pr[\Pi   (r) = r' | T' = \tau   ']}

\dcontx{= \frac{\pr[\Delta   (T) = \tau   ' \circ    \Pi    \land    \Pi   (r) = r']}{\pr[\Delta   (T) = \tau   ' \circ    \Pi   ] }}
{from \refeq{delta_pi}}

\dcont{= \frac{\dsum   {\tau   \in{\cal T}}{f_{T}(\tau   ) \pr[\Delta   (\tau   )= \tau   ' \circ    \Pi    \land    \Pi   (r) = r']}}{\dsum   {\tau\in{\cal T}}{f_{T}(\tau   ) \pr[\Delta   (\tau   )= \tau   ' \circ    \Pi   ]}}}

Next we show that $f_{T}$ maximizes the above estimation probability.
In other words, we show 
which adversary can guess the record of a person with the highest confidence.

\begin{lemma}\label{lm:3}
Let $\mathbb{R}^{n+}$ be the set of non-zero $n$-dim vectors whose elements are non-negative real numbers. 
Then for any vector $a, b \in   \mathbb{R}^{n+}$, 
the maximum of
$$g(x) \defeq \disp{\frac{b \cdot x}{a \cdot x} (= \frac{\sum   _{i < n}{b_{i} x_{i}}}{\sum   _{i < n}{a_{i} x_{i}}})}$$
on a variable $x$ on $\mathbb{R}^{n+}$
is 
$\dmax{i < n}{\frac{b_{i}}{a_{i}}}$, 
and $x$ satisfies
\begin{equation*}
\text{for any } i < n \text{ such that } \dfrac{b_{i}}{a_{i}} \neq    \dmax{i < n}{\dfrac{b_{i}}{a_{i}}}, x_{i} = 0.
\end{equation*}
\end{lemma}

\noindent{({\it proof of \reflm{3}})}\\
Since $g(x)$ is invariant on a scalar multiplication of $x$,
it is sufficient to find the maximum in some $Y \subset    \mathbb{R}^{n+}$ such that 
there exist $\alpha    \in    \mathbb{R}$ and $y \in    Y$ that satisfy $\alpha    y = x$, for $x \in    \mathbb{R}^{n+}$.
By taking $Y$ as a plane, 
we can find that the maximum exists because it is a bounded closed set.  

Next we have that
$$x_{i} = 0 \text{ or }\dfrac{\partial g(x)}{\partial x_{i}} = 0$$
holds
for each element $x_{i}$ of $x\in  \mathbb{R}^{n+}$ that gives maximum $g(x)$.
Otherwise, escalating $x_{i}$ should increase the value of $g(x)$, 
and contradicts that $g(x)$ is the maximum. 
Because of this fact and also because that $x$ is not a zero vector, 
there must exist at least one $i$ such that $\dfrac{\partial g(x)}{\partial x_{i}} = 0$. 

Finally, this partial differential is found to be
$$\dfrac{\partial g(x)}{\partial x_{i}} = \dfrac{(a \cdot x)b_{i} - (b \cdot x)a_{i}}{(a \cdot x)^{2}},$$
then
$$\dfrac{\partial g(x)}{\partial x_{i}} = 0 \Leftrightarrow    g(x) = \dfrac{b_{i}}{a_{i}}$$
holds. 
Therefore, $i$, which satisfies $\dfrac{\partial g(x)}{\partial x_{i}} = 0$ must be $i$ giving maximum $\dfrac{b_{i}}{a_{i}}$; 
all other elements are $0$, 
and the maximum of $g(x)$ is $\dmax{i < n}{\frac{b_{i}}{a_{i}}}$. 
\begin{flushright}
$\Box   $(\reflm{3})
\end{flushright}

From the above lemma, 
when ${\cal E}   (f_{T}, \tau   ', r, r')$ takes the maximum, 
$f_{T}$ makes the following formula maximum, 
\begin{equation}
\dfrac{\pr[\Delta   (\tau   )= \tau   ' \circ    \Pi    \land    \Pi   (r) = r')]}{\pr[\Delta   (\tau   )= \tau   ' \circ    \Pi   ]} \label{fm:pr1}
\end{equation}
and the maximum of \reffm{pr1} is equal to that of ${\cal E}   (f_{T}, \tau   ', r, r')$. 

Since $\Pi  $ is a uniformly random permutation, \reffm{pr1} is transformed as follows.

\dhead{\reffm{pr1} = \dfrac{\dfrac{1}{|{\cal R}|!}\dsum   {\pi   (r) = r'}{\pr[\Delta   (\tau   ) = \tau   ' \circ    \pi   ]}}{\dfrac{1}{|{\cal R}|!}\dsum   {\pi   }{\pr[\Delta   (\tau   ) = \tau   ' \circ    \pi   ]}}}

\dcontx{= \dfrac{\dsum   {\pi   (r) = r'}{\pr[\Delta   (\tau   ) = \tau   ' \circ    \pi   ]}}{\dsum   {\pi   }{\pr[\Delta   (\tau   ) = \tau   ' \circ    \pi   ]}}= \dfrac{\dsum   {\pi   (r) = r'}{\dprod   {s \in    {\cal R}   }{\pr[(\Delta   (\tau   ))(s) = \tau   '(\pi   (s))]}}}{\dsum   {\pi   }{\dprod   {s \in    {\cal R}   }{\pr[(\Delta   (\tau   ))(s) = \tau   '(\pi   (s))]}}}}
{since $\Delta  $ is independent from each record}

Let a matrix $A^{\tau, \tau'   }$ be 
$$A^{\tau, \tau'   }_{s, s'} \defeq \pr[(\Delta   (\tau   ))(s) = \tau   '(s')]$$
for any $s \in    {\cal R}   , s' \in    {\cal R}   '$.
Then, the above formula is represented as follows. 

\[
F(A^{\tau}) \defeq \dfrac{\dsum   {\pi   (r) = r'}{\dprod   {s \in    {\cal R}   }{A^{\tau, \tau'   }_{s, \pi   (s)}}}}{\dsum   {\pi   }\dprod   {s \in    {\cal R}   }{{A^{\tau, \tau'   }_{s, \pi   (s)}}}}
\]




We would rather find the minimum of the reciprocal
than the maximum of $F(A^{\tau, \tau'})$ itself. 
In the case of $|{\cal R}| \geq 2$,
the reciprocal is transformed as follows.
 
\dhead{\dfrac{1}{F(A^{\tau, \tau'})} = \dfrac{\dsum{\pi}\dprod{s \in {\cal R}}{{A^{\tau, \tau'}_{s, \pi(s)}}}}{\dsum{\pi(r) = r'}{\dprod{s \in  {\cal R}}{A^{\tau, \tau'}_{s, \pi(s)}}}}}

\dcont{= \dfrac{\dsum{t\neq r \atop t'\neq r'}A^{\tau, \tau'}_{t, r'}A^{\tau, \tau'}_{r, t'}\dsum{\pi(t)=r' \atop \pi  (r)=t'}\dprod{s\neq t,r}A^{\tau, \tau'}_{s, \pi(s)} + A^{\tau, \tau'}_{r, r'}\dsum{\pi(r) = r'}\dprod{s\neq r}A^{\tau, \tau'}_{s, \pi(s)}}{A^{\tau, \tau'}_{r, r'}\dsum{\pi(r) = r'}\dprod{s\neq r}A^{\tau, \tau'}_{s, \pi(s)}}}

\dcont{= 1 + \dfrac{\dsum{t\neq r \atop t'\neq r'}A^{\tau, \tau'}_{t, r'}A^{\tau, \tau'}_{r, t'}\dsum{\pi(t)=r' \atop \pi(r)=t'}\dprod{s\neq t,r}A^{\tau, \tau'}_{s, \pi(s)}}{A^{\tau, \tau'}_{r, r'}\dsum{\pi(r) = r'}\dprod{s\neq r}A^{\tau, \tau'}_{s, \pi(s)}}
= 1 + \dfrac{\dsum{t\neq r \atop t'\neq r'}A^{\tau, \tau'}_{t, r'}A^{\tau, \tau'}_{r, t'}\dsum{\pi(r)=r' \atop \pi(t)=t'}\dprod{s\neq t,r}A^{\tau, \tau'}_{s, \pi(s)}}{A^{\tau, \tau'}_{r, r'}\dsum{\pi(r) = r'}\dprod{s\neq r}A^{\tau, \tau'}_{s, \pi(s)}}
}

\dcont{= 1 + \dfrac{\dsum{t\neq r}A^{\tau, \tau'}_{t, r'}A^{\tau, \tau'}_{r, \pi(t)}\dsum{\pi(r)=r'}\dfrac{\dprod{s\neq r}A^{\tau, \tau'}_{s, \pi(s)}}{A^{\tau, \tau'}_{t, \pi(t)}}}{A^{\tau, \tau'}_{r, r'}\dsum{\pi(r) = r'}\dprod{s\neq r}A^{\tau, \tau'}_{s, \pi(s)}}
= 1 + \dfrac{\dsum{\pi(r)=r'}\dsum{t\neq r}\dfrac{A^{\tau, \tau'}_{t, r'}A^{\tau, \tau'}_{r, \pi(t)}}{A^{\tau, \tau'}_{t, \pi(t)}}\dprod{s\neq r}A^{\tau, \tau'}_{s, \pi(s)}}{\dsum{\pi(r) = r'}A^{\tau, \tau'}_{r, r'}\dprod{s\neq r}A^{\tau, \tau'}_{s, \pi(s)}}
}

We show the following lemma.
\begin{lemma}\label{lm:4}
Let $g_{i}$ and $h_{i}$ be $g_{i}, h_{i} : \mathbb{R}^{{\cal I}} \to \mathbb{R}$ for any index $i\in{\cal I}$, where ${\cal I}$ is a set of indices.
If some $x\in\mathbb{R}^{{\cal I}}$ and $z\in\mathbb{R}$ satisfy $\dfrac{h_{i}(x)}{g_{i}(x)} = \dmin{x'\in\mathbb{R}^{{\cal I}}}{\dfrac{h_{i}(x')}{g_{i}(x')}} = z$ for any $i\in{\cal I}$, 
then the following equation is satisfied.
$$\dmin{x'\in \mathbb{R}^{{\cal I}}}{\dfrac{\dsum{i\in{\cal I}}{h_{i}(x')}}{\dsum{i\in{\cal I}}{g_{i}(x')}}} = \dfrac{\dsum{i\in{\cal I}}{h_{i}(x)}}{\dsum{i\in{\cal I}}{g_{i}(x)}}$$
\end{lemma}

\noindent({\it proof of \reflm{4}})\\
From the assumption of the lemma, $h_{i}(x') \geq z g_{i}(x')$ hold for all $i\in{\cal I}$ and any $x'\in\mathbb{R}$.
Therefore, 
$$\dfrac{\dsum{i\in{\cal I}}{h_{i}(x')}}{\dsum{i\in{\cal I}}{g_{i}(x')}} \geq z,$$ 
then 
$$\dmin{x'\in \mathbb{R}^{n}}{\dfrac{\dsum{i\in{\cal I}}{h_{i}(x')}}{\dsum{i\in{\cal I}}{g_{i}(x')}}} = \dfrac{\dsum{i\in{\cal I}}{h_{i}(x)}}{\dsum{i\in{\cal I}}{g_{i}(x)}}$$
holds.
\begin{flushright}
$\Box   $(\reflm{4})
\end{flushright}

Let $h_{\pi}(A^{\tau, \tau'})$ and $g_{\pi}(A^{\tau, \tau'})$ be 
$$h_{\pi}(A^{\tau, \tau'}) = \dsum{t\neq r}\dfrac{A^{\tau, \tau'}_{t, r'}A^{\tau, \tau'}_{r, \pi(t)}}{A^{\tau, \tau'}_{t, \pi(t)}}\dprod{s\neq r}A^{\tau, \tau'}_{s, \pi(s)},$$
$$g_{\pi}(A^{\tau, \tau'}) = A^{\tau, \tau'}_{r, r'}\dprod{s\neq r}A^{\tau, \tau'}_{s, \pi(s)}$$
for any $\pi:{\cal R}\to{\cal R}'$.
Thanks to \reflm{4}, it is sufficient to consider $\dfrac{h_{\pi}(A^{\tau, \tau'})}{g_{\pi}(A^{\tau, \tau'})}$
only. 
Because it is transformed into
$\dfrac{1}{A^{\tau,\tau'}_{r, r'}}\dsum{t\neq r}\dfrac{A^{\tau, \tau'}_{t, r'}A^{\tau,\tau'}_{r, \pi(t)}}{A^{\tau,\tau'}_{t, \pi(t)}}$, 
it takes the minimum for any $\pi:{\cal R}\to{\cal R}'$ 
when $\tau$ and $\tau'$ are as follows. 
\begin{quote}
There exists $v\in{\cal V}$ and $v'\in{\cal V}'$ such that $\dfrac{A_{v, \tau'(r')}A_{\tau(r), v'}}{A_{\tau(r), \tau'(r')}A_{v, v'}} =\dmin{u, v\in{\cal V} \atop u',v' \in{\cal V}'}\dfrac{A_{u, v'}A_{v, u'}}{A_{u, u'}A_{v, v'}}$, $\tau(s) = v$ for any $s\neq r$ and $\tau'(s') = v'$ for any $s'\neq r'$.
\end{quote}

Since $k$ is to be the reciprocal of the maximum of $F(A^{\tau, \tau'})$,  
$k$ is found to be the following value.
\begin{eqnarray*}
k &=& 1 + (|{\cal R}| - 1)\dmin{u, v\in{\cal V} \atop u',v' \in{\cal V}'}\dfrac{A_{u, v'}A_{v, u'}}{A_{u, u'}A_{v, v'}}
\end{eqnarray*}

It is easy to confirm that the above equation also holds when $|{\cal R}| = 1$.
In this case, since only one $\pi$ exists (denoted as $\hat{\pi}$), $k$ equals $1$ as follows.

\[
k = \dfrac{1}{F(A^{\tau, \tau'})} = \dfrac{\dsum{\pi}\dprod{s \in {\cal R}}{{A^{\tau, \tau'}_{s, \pi(s)}}}}{\dsum{\pi(r) = r'}{\dprod{s \in  {\cal R}}{A^{\tau, \tau'}_{s, \pi(s)}}}}
= \dfrac{\dprod{s \in {\cal R}}{{A^{\tau, \tau'}_{s, \hat{\pi}(s)}}}}{{\dprod{s \in  {\cal R}}{A^{\tau, \tau'}_{s, \hat{\pi}(s)}}}}
= 1
 = 1 + (|{\cal R}| - 1)\dmin{u, v\in{\cal V} \atop u',v' \in{\cal V}'}\dfrac{A_{u, v'}A_{v, u'}}{A_{u, u'}A_{v, v'}}
\]
\hfill {(since $|{\cal R}| = 1$)}
\begin{flushright}
$\Box   $
\end{flushright}

We describe the multi-attribute version of \refth{2}. 
\begin{corollary}
A PRAM mechanism whose transition probability matrices are $A_{a}$ for each attribute $a$ is a $Pk$-anonymization when $k$ is described as follows, 
\begin{equation*}
k = 1 + (|{\cal R}| - 1) \dprod{a \in {\cal A}}{\rm AR}_{a}
\end{equation*}
where ${\rm AR}_{a}$ is
$${\rm AR}_{a} = \dmin{u, v\in{\cal V} \atop u', v' \in{\cal V}'}\dfrac{(A_{a})_{u, v'}(A_{a})_{v, u'}}{(A_{a})_{u, u'}(A_{a})_{v, v'}}.$$
\end{corollary}

The following corollary is applicable to retention-replacement perturbation. 
\begin{corollary}
Retention-replacement perturbation  whose retention probabilities are $\rho_{a}$ for each attribute $a\in{\cal A}$, 
is a $Pk$-anonymization when $k$ is described as follows, 
\begin{equation*}
k = 1 + (|{\cal R}| - 1) \dprod{a \in {\cal A}}{\rm AR}_{a}
\end{equation*}
where ${\rm AR}_{a}$ is 
$${\rm AR_{a}} = \left({\dfrac{1 - \rho_{a}}{1 + (|{\cal V}_{a}| - 1)\rho_{a}}}\right)^{2}.$$
\end{corollary}

Using \refth{2}, $k$ is easily calculated with the record count $|{\cal R}|$ and transition probability matrix $A$. 
Regarding retention-replacement perturbation, 
$A$ is determined independently with the instance of private data, 
$k$ is calculated with the record count $|{\cal R}|$ and the numbers of attribute values $|{\cal V}_{a}|$ and retention probabilities $\rho_{a}$ only, for each attribute $a$. 

Conversely, $\rho_{a}$ are also calculated in retention-replacement perturbation. 
By letting all $\rho_{a}$ be the same $\rho$ over all attributes, the equation is transformed as follows. 
\begin{equation}\label{eq:rrp}
k = 1 + (|{\cal R}| - 1) \left(\dprod{a \in {\cal A}}{\frac{1 - \rho}{1 + (|{\cal V}_{a}| - 1)\rho}})\right)^{2}
\end{equation}
Since $k$ monotonically decreases on $0 \leq\rho\leq 1$, 
$\rho$ is easily and uniquely solved using, for example, the bisection method for any $k$, $|{\cal R}|$ and $|{\cal V}_{a}|$(\refalg{rho_from_k}).
Note that $k$ is allowed to be a real number, 
for example, $k = 1.5$. 

\begin{algorithm}
\caption{determining $\rho$ in retention-replacement perturbation from $k$\newline
input: $k\in\mathbb{R}(k \geq 1)$, $|{\cal R}|\in\mathbb{N}$, $|{\cal V}_{a}|$ for each attribute\newline
output: retention probability $\rho$}
\label{alg:rho_from_k}
\begin{algorithmic}[1]
\st Set $\rho_{0} = 1/2$. 
\st Run the bisection method with $\rho$'s initial value $\rho_{0}$ with respect to $k$ using \refeq{rrp} and output the converged $\rho$. 
\end{algorithmic}
\end{algorithm}

For example, 
to ensuring $P100$-anonymity on $100,000$ records of data, 
$\rho  $ is calculated as roughly $0.303$,
 where there are three attributes, 
sex, age from $20$'s to $60$'s, and $10$-leveled annual income. 

When the record count is uncertain since the data are to be collected thereafter,
it is sufficient to use the expected record count.
Even when the record count does not reach the expected value, 
$Pk$-anonymity is still satisfied for the following reason.
When each record in table $\tau'$ is anonymous due to an anonymous communication channel,
it can be said that only a part of table $\tau'$ is visible in the state in which $\tau'$ is being collected.
An estimation in such a situation is equivalent to that from the algorithm that ignores the absent records.
From Proposition~1, the algorithm cannot derive ${\cal E}  (f_{T}, \tau  ', r, r') \gneq  1/k$ if it is correct.

\subsection{DP on PRAM in Addition to Pk-Anonymity}

Regarding retention-replacement perturbation, 
we can derive \refalg{rho_from_e} that 
determines the parameter in order to satisfy $\varepsilon$-DP
from \refcor{dp3}.

\begin{algorithm}
\caption{determining $\rho$ from $\varepsilon$\newline
input: $\varepsilon>0$ and $|{\cal V}_{a}|$ for each attribute $a$\newline
output: retention probability $\rho$}
\label{alg:rho_from_e}
\begin{algorithmic}[1]
\st Set $\rho_{0} = 1/2$. 
\st Run the bisection method with $\rho$'s initial value $\rho_{0}$ with respect to $\varepsilon$ using \refeq{rrp_dp}, and output the converged $\rho$. 
\end{algorithmic}
\end{algorithm}

Combining \refalg{rho_from_e} with \refalg{rho_from_k},
we have \refalg{rho_from_ke} that 
determines the parameter in order to satisfy both
$Pk$-anonymity and $\varepsilon$-DP.

\begin{algorithm}
\caption{determining $\rho$ from $k$ and $\varepsilon$\newline
input: $k\in\mathbb{R}(k \geq 1)$, $\varepsilon>0$, $|{\cal R}|\in\mathbb{N}$, $|{\cal V}_{a}|$ for each attribute\newline
output: retention probability $\rho$}
\label{alg:rho_from_ke}
\begin{algorithmic}[1]
\st Run \refalg{rho_from_k} and \refalg{rho_from_e} and let the results be $\rho_{k}$ and $\rho_{\varepsilon}$, respectively. 
\st output $\min(\rho_{k}, \rho_{\varepsilon})$. 
\end{algorithmic}
\end{algorithm}

\section{Experimental results}\label{sec:ex}

From the aspect of utility,
we show that randomized data-bases protected by $Pk$-anonymity are available for data analyses.
We experimented with cross-tabulations (or, contingency tables) using $Pk$-anonymity. 

In the experiments discussed below, 
the dataset was randomized  by retention-replacement perturbation, 
and cross tabulations were calculated using the reconstruction method \cite{AST05}.
The target dataset was the US census dataset in the UCI Machine Learning Repository \cite{BM98}, which has $2,458,285$ records.
Out of this dataset, we extracted and used $7$ attributes, as shown on \reftbl{atts}.
Several attributes were rounded because they had too many attribute values for cross tabulation. 

\begin{table}[h]
 \begin{center}{
\caption{attributes and number of attribute values}
\label{tbl:atts}
\begin{tabular}{|c|c|}\hline
Sex     & 2 \\ \hline
Age[*]      &   18 \\ \hline
Total Pers. Inc. Signed[*] & 12 \\ \hline
Worked Last Yr. 1989 & 3 \\ \hline
Worked Last Week & 3 \\ \hline
Ed. Attainment & 18 \\ \hline
Travel Time to Work[*] & 20 \\ \hline
\end{tabular}\\
(Marked([*]) attributes are rounded.)
 }\end{center}
\end{table}

\reffig{recs} shows $L1$-norm errors and $\epsilon$ by varying the record count with fixed $k = 2$ and four attributes, 
Sex, Age, Total Pers. Inc. Signed, and Worked Last Yr. 1989.
The $L1$-norm is a normalized distance between original cross-tabulated aggregates and reconstructed aggregates, 
given as the following $d$, where each $x_{v}$ and $y_{v}$ are the counted aggregates of the private table and the reconstructed aggregates corresponding to $v\in{\cal V}$, respectively.
$$d = \dfrac{\dsum{v \in {\cal V}}|x_{v} - y_{v}|}{|{\cal R}|}$$
From the graph, it seems that errors become smaller as the record count increases.
When only $245$ records were used, errors were quite high.
However, there were almost no errors when all $2,458,285$ records were used. 
This is due to two reasons.
First, in the reconstruction method, a large number of records generally results in accurate analyses results in a fixed retention probability. 
Second, since many records also provide high $k$ on $Pk$-anonymity by the same $\rho$ according to \refth{2}, 
one can set a relatively high $\rho$. 
Regarding $\varepsilon$-DP, $\varepsilon$ increases as the record count increases. 
It is because retention probability $\rho$ monotonically increases with the increase of the record count by \refeq{rrp}, when $k$ is fixed; 
\begin{figure}[htbp]
\centering
\includegraphics[width = 9cm]{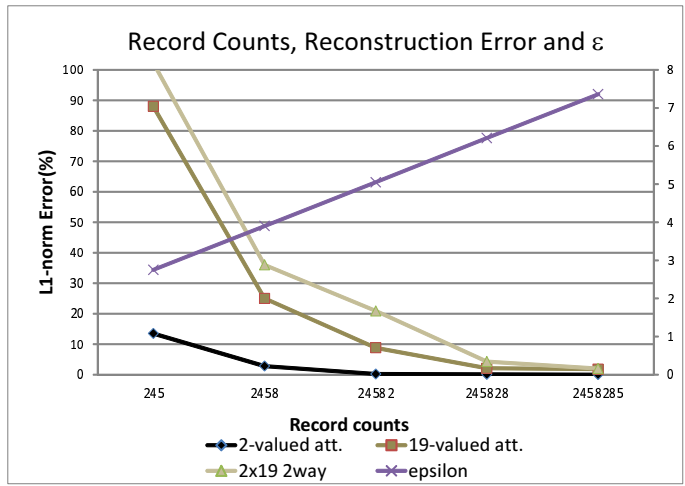}	
\caption{Reconstruction errors and $\epsilon$ by varying number of records} \label{fig:recs}
\end{figure}

\reffig{atts} shows $L1$-norm errors and $\epsilon$ by varying the number of attributes with fixed $k = 2$, 
using all the records of the dataset. 
Attributes have been added in the same order as in \reftbl{atts}.
\reffig{k} shows $L1$-norm errors and $\epsilon$ by varying $k$ with fixed attributes. 
All the records from the dataset were used and attributes were the same four attributes as in \reffig{recs}.
From these graphs, it seems that errors become larger as the number of attributes or $k$ increases.
However, the increment is quite small.

\begin{figure}[htbp]
\centering
\includegraphics[width = 9cm]{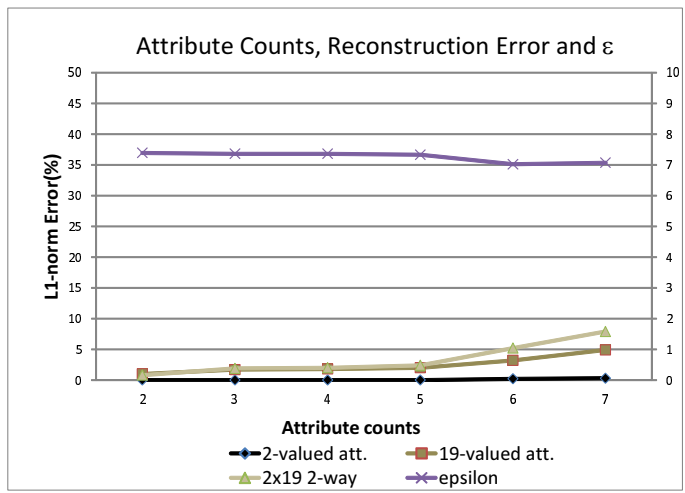}	
\caption{Reconstruction errors and $\epsilon$ by varying number of attributes} \label{fig:atts}p
\end{figure}

\begin{figure}[htbp]
\centering
\includegraphics[width = 9cm]{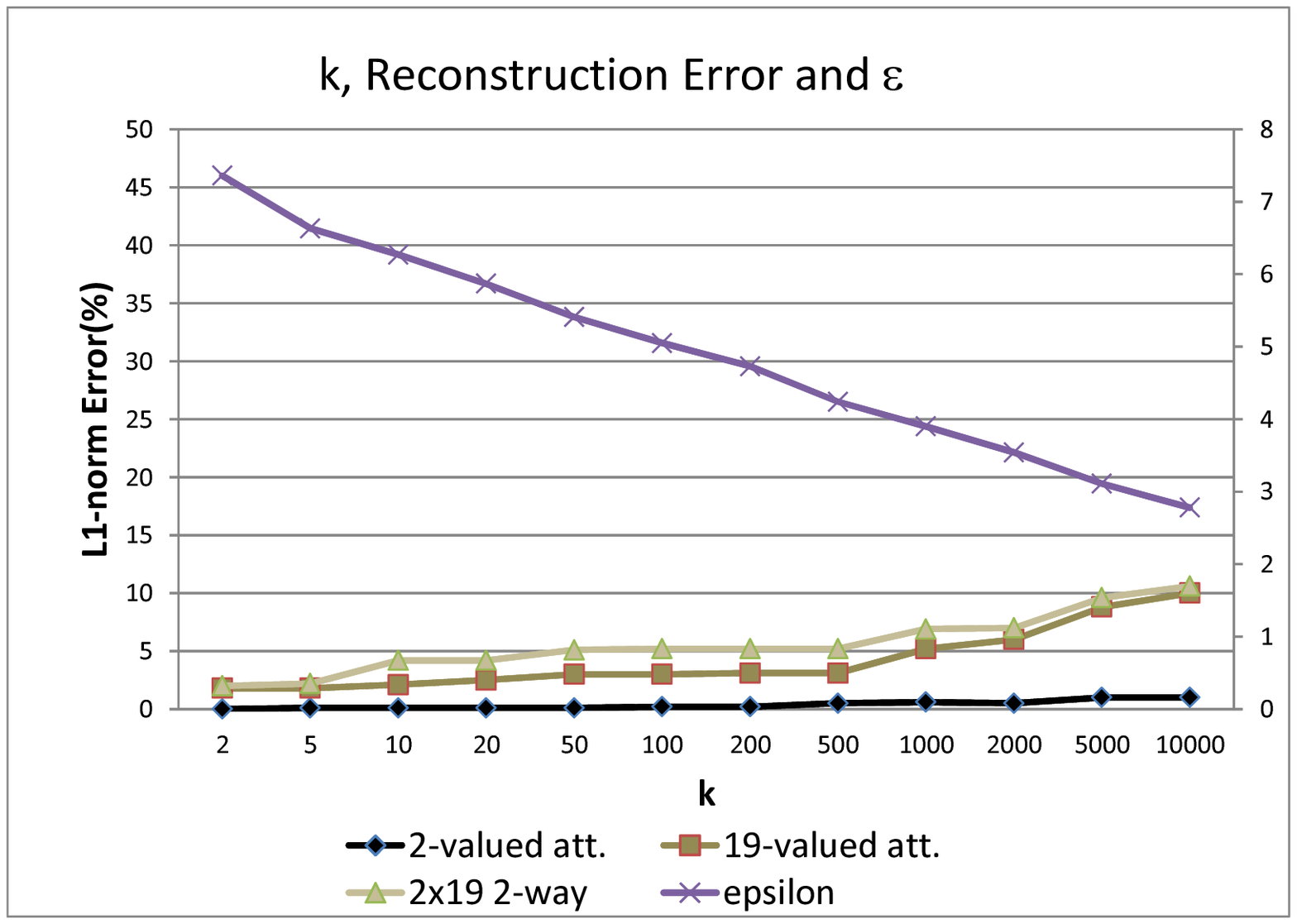}	
\caption{Reconstruction errors and $\epsilon$ by varying $k$} \label{fig:k}
\end{figure}




\section{Conclusions}\label{sec:conclusions}
In the field of anonymized microdata release, 
we mainly presented the following two theories.  
We first proposed an anonymity notion, $Pk$-anonymity, 
which is an extension of $k$-anonymity to randomized microdata, 
and its intuitive meaning is 
\dquote{no one estimates which person the record came from with more than $1/k$ probability.}
We then applied $Pk$-anonymity to PRAM. 
PRAM is known to satisfy $\varepsilon$-DP; 
thus, 
it achieves $k$-anonymous and $\varepsilon$-differentially private microdata release. 

The contributions of the paper are:
\begin{itemize}
\item an anonymity notion called $Pk$-anonymity,
\item proofs that $Pk$-anonymity is an exact mathematical extension of $k$-anonymity,
\item a formula for calculating $k$ on PRAM,
\item algorithms for determining the parameter of the retention-replacement perturbation according to $k$ and $\varepsilon$,
\item experimental results to empirically analyze the trade-off relation between utility and privacy/anonymity using a real dataset. 
\end{itemize}

Theoretical analyses and further experiments in real applications regarding utility are future work.




%
\bibliographystyle{abbrv}
\bibliography{smc,reconst,k_anonymity,dp,db}  
%
%

\appendix 

\section{Proof of Theorem 1}\label{app:dp}
\label{ap:thdp1}

Let $\tau_{1}, \tau_{2}\in{\cal T}$ be arbitrary private tables differing by one record
(i.e., one and only one $r \in {\cal R}$ exists and satisfies [$\tau_{1}(r) \neq \tau_{2}(r)$]
and $\tau_{1}|_{{\cal R}\setminus\{ r\}} = \tau_{2}|_{{\cal R}\setminus\{ r\}}$)
and let $\tau' \in {\cal T}'$ be an arbitrary randomized table. 

The proposition we should show is the following inequality.  
$$\dmax{\tau_{1}, \tau_{2}, \tau'}\dfrac{\pr[\Delta(\tau_{1}) = \tau']}{\pr[\Delta(\tau_{2}) = \tau']} 
\leq \exp(\varepsilon) =\dmax{u, v \in {\cal V} \atop v' \in {\cal V}'}\dfrac{A_{u, v'}}{A_{v, v'}}$$

The left-hand side of the above inequality is transformed as follows. 

\begin{eqnarray}
&&\dmax{\tau_{1}, \tau_{2}, \tau'}\dfrac{\pr[\Delta(\tau_{1}) = \tau']}{\pr[\Delta(\tau_{2}) = \tau']} \nonumber\\
&=& \dmax{\tau_{1}, \tau_{2}, \tau'}\dfrac{\dsum{\pi:{\cal R}\to{\cal R}}\dfrac{1}{|{\cal R}|}\dprod{s\in{\cal R}}\pr[(\Delta(\tau_{1}))(s) = \tau'(\pi(s))]}
      {\dsum{\pi:{\cal R}\to{\cal R}}\dfrac{1}{|{\cal R}|}\dprod{s\in{\cal R}}\pr[(\Delta(\tau_{2}))(s) = \tau'(\pi(s))]} \nonumber\\
&=& \dmax{\tau_{1}, \tau_{2}, \tau'}\dfrac{\dsum{\pi:{\cal R}\to{\cal R}}\dprod{s\in{\cal R}}A_{\tau_{1}(s), \tau'(\pi(s))}}
      {\dsum{\pi:{\cal R}\to{\cal R}}\dprod{s\in{\cal R}}A_{\tau_{2}(s), \tau'(\pi(s))}} \nonumber\\
&=& \dmax{\tau_{1}, \tau_{2}, \tau'}\dfrac{\dsum{\pi:{\cal R}\to{\cal R}}(\dprod{s\neq r}A_{\tau_{1}(s), \tau'(\pi(s))})A_{\tau_{1}(r), \tau'(\pi(r))}}
      {\dsum{\pi:{\cal R}\to{\cal R}}(\dprod{s\neq r}A_{\tau_{2}(s), \tau'(\pi(s))})A_{\tau_{2}(r), \tau'(\pi(r))}} \label{fm:dp1}
\end{eqnarray}

Now, 
let $v_{1}$, $v_{2}$ and $\bar{\tau}$ be $\tau_{1}(r)$, $\tau_{2}(r)$ and 
$\tau_{1}|_{{\cal R}\setminus\{ r\}} (= \tau_{2}|_{{\cal R}\setminus\{ r\}})$, 
respectively. 
Note that these three variables determine $\tau_{1}$ and $\tau_{2}$ uniquely. 
Furthermore, 
let $a_{\pi, v_{2}, \tau'}$, $b_{\pi, v_{1}, \tau'}$ and $x_{\pi, \bar{\tau}, \tau'}$ denote 
$A_{\tau_{2}(r), \tau'(\pi(r))}$, $A_{\tau_{1}(r), \tau'(\pi(r))}$ and \\
$\dprod{s\neq r}A_{\tau_{1}(s), \tau'(\pi(s))} (=\dprod{s\neq r}A_{\tau_{2}(s), \tau'(\pi(s))})$, 
respectively. 
Using these representations, 
we denote the above formula with the following formula. 

\begin{eqnarray}
(\ref{fm:dp1})
&=&\dmax{\tau_{1}, \tau_{2}, \tau'}\dfrac{\dsum{\pi:{\cal R}\to{\cal R}}b_{\pi, v_{1}, \tau'}x_{\pi, \bar{\tau}, \tau'}}
      {\dsum{\pi:{\cal R}\to{\cal R}}a_{\pi, v_{2}, \tau'}x_{\pi, \bar{\tau}, \tau'}} \nonumber\\
&=& \dmax{\bar{\tau}, v_{1}, v_{2}, \tau'}\dfrac{\dsum{\pi:{\cal R}\to{\cal R}}b_{\pi, v_{1}, \tau'}x_{\pi, \bar{\tau}, \tau'}}
      {\dsum{\pi:{\cal R}\to{\cal R}}a_{\pi, v_{2}, \tau'}x_{\pi, \bar{\tau}, \tau'}} \label{fm:dp2}
\end{eqnarray}

By fixing $v_{1}$, $v_{2}$ and $\tau'$, 
we can apply \reflm{3} to remove $\bar{\tau}$ and $x_{\pi, \bar{\tau}, \tau'}$ from the above maximum.  

\begin{eqnarray}
(\ref{fm:dp2})
&=& \dmax{v_{1}, v_{2}, \tau', \pi}\dfrac{b_{\pi, v_{1}, \tau'}}
      {a_{\pi, v_{2}, \tau'}} \nonumber\\
&\leq& \dmax{v_{1}, v_{2}, \tau', \pi}\dfrac{A_{\tau_{1}(r), \tau'(\pi(r))}}
      {A_{\tau_{2}(r), \tau'(\pi(r))}}
    \leq \dmax{u, v \in {\cal V} \atop v' \in {\cal V}'}\dfrac{A_{u, v'}}{A_{v, v'}}
\end{eqnarray}
\begin{flushright}
$\Box   $
\end{flushright}

\ignore{

\section{Proof of Theorem 4.5}\label{ap:thdp2}

We use the following additional notations in this proof. 
\begin{itemize}
\item $\mathscr{T}, \mathscr{T}'$: the set of private/released tables whose record count is arbitrary
\item $N_{\tau}$: the record count of a table $\tau$
\item $\mathscr{T}_{n}$: the set of private tables whose record count is $n\in\mathbb{N}$ or more
\item $\mathscr{L}$: the set of any person 
\item ${\cal R}_{\tau}$: the set of all records in table $\tau$ 
\end{itemize}

Let $\tau_{1}\in\mathscr{T}_{n}$, 
$\tau_{2}$ be a table generated from $\tau_{1}$ by inserting one arbitrary record, 
and $\tau' \in \mathscr{T}'$ be an arbitrary randomized table.  
In addition, for any table $\tau\in\mathscr{T}_{n}$, we denote the record count of $\tau$ by $N_{\tau}$, 
and let $d$ be a value on discrete double exponential distribution in \refalg{dde}.  

The propositions to prove are the following two inequalities, 
where $\eta = \dfrac{1 + p - p^{n + 1}}{1 + p - p^{n}}p^{-1}\left(\dmax{u, v \in {\cal V} \atop v' \in {\cal V}'}\dfrac{A_{u, v'}}{A_{v, v'}}\right)$. 
\begin{eqnarray*}
\dmax{\tau_{1}, \tau_{2}, \tau'}\dfrac{\pr[\Delta\circ{\cal D}_{p}(\tau_{1}) = \tau']}{\pr[\Delta\circ{\cal D}_{p}(\tau_{2}) = \tau']} \leq \eta\\
\dmax{\tau_{1}, \tau_{2}, \tau'}\dfrac{\pr[\Delta\circ{\cal D}_{p}(\tau_{2}) = \tau']}{\pr[\Delta\circ{\cal D}_{p}(\tau_{1}) = \tau']} \leq \eta
\end{eqnarray*}
We prove the former inequality. 
The latter is proven similarly. 

(1) when $N_{\tau'} \leq N_{\tau_{1}}$:\\

For each $\tau_{i}$ to satisfy $\Delta\circ{\cal D}_{p}(\tau_{i}) = \tau'$, 
$D$ must be $N_{\tau'} - N_{\tau_{i}}$.  
For $i = 1, 2$, 
by letting the set of all injective maps from ${\cal R}_{\tau_{i}}$ to ${\cal R}_{\tau'}$ be denoted as
$\Pi_{i}$,  
we obtain the following transformation.  
\begin{eqnarray*}
&&  \pr[\Delta\circ{\cal D}(\tau_{i}) = \tau'] \\
&=& \pr[D = N_{\tau'} - N_{\tau_{i}}]\dsum{\pi\in\Pi_{i}}\dfrac{\dprod{s'\in{\cal R}_{\tau'}}A_{\pi^{-1}(s'), \tau'(s')}}{\perm{N_{\tau_{i}}}{N_{\tau'}}}\\
\end{eqnarray*}
By using this, 
we can decompose the maximum of the proportion of probabilities into two maxima as follows.
\begin{eqnarray*}
&&\dmax{\tau_{1}, \tau_{2}, \tau'}\dfrac{\pr[\Delta\circ{\cal D}_{p}(\tau_{1}) = \tau']}{\pr[\Delta\circ{\cal D}_{p}(\tau_{2}) = \tau']}\\
&\leq& \dfrac{N_{\tau_{2}} - N_{\tau'}}{N_{\tau_{2}}}
\dmax{\tau_{1}, \tau_{2},\tau'}\dfrac{\pr[D = N_{\tau'} - N_{\tau_{1}}]}{\pr[D = N_{\tau'} - N_{\tau_{2}}]}\\
&\times& \dmax{\tau_{1}, \tau_{2},\tau'}\dfrac{\dsum{\pi\in\Pi_{1}}\dprod{s'\in{\cal R}_{\tau'}}A_{\pi^{-1}(s'), \tau'(s')}}
{\dsum{\pi\in\Pi_{2}}\dprod{s'\in{\cal R}_{\tau'}}A_{\pi^{-1}(s'), \tau'(s')}}
\end{eqnarray*}
We try to bound these two maxima. 
The following equation holds for $i = 1, 2$. 
$$\pr[D = N_{\tau'} - N_{\tau_{i}}] = \dfrac{1}{1- \dfrac{p^{N_{\tau_{i}}}}{(1 + p)}}\cdot\dfrac{1 - p}{1 + p}p^{N_{\tau'} - N_{\tau_{i}}}$$
Note that $\pr[D = N_{\tau'} - N_{\tau_{i}}] \neq \dfrac{1 - p}{1 + p}p^{N_{\tau'} - N_{\tau_{i}}}$
because $D + N_{\tau_{1} }$ is restricted to be non-negative due to step 2 in \refalg{dde}.
This equation allows us to find the first maximum.   
$$\dmax{\tau_{1}, \tau_{2},\tau'}\dfrac{\pr[D = N_{\tau'} - N_{\tau_{1}}]}{\pr[D = N_{\tau'} - N_{\tau_{2}}]}
= \dfrac{1 + p - p^{n + 1}}{1 + p - p^{n}}p^{-1}$$
For the second maximum, 
by applying \reflm{3} twice, 
we have its bound. 
$$\dmax{\tau_{1}, \tau_{2},\tau'}\dfrac{\dsum{\pi\in\Pi_{1}}\dprod{s'\in{\cal R}_{\tau'}}A_{\pi^{-1}(s'), \tau'(s')}}
{\dsum{\pi\in\Pi_{2}}\dprod{s'\in{\cal R}_{\tau'}}A_{\pi^{-1}(s'), \tau'(s')}}
\leq \dmax{u, v \in {\cal V} \atop v' \in {\cal V}'}\dfrac{A_{u, v'}}{A_{v, v'}}$$
Note that this bound is loose; 
however, 
it is sufficient to prove this theorem. 

(2) when $N_{\tau'} > N_{\tau_{1}}$:\\

As in (1), 
we can decompose the maximum into two maxima. 
The first maximum is exactly the same as in (1), 
and the second maximum is reduced to the maximum in the proof of \refth{dp2} by applying \reflm{3} twice.  
\begin{flushright}
$\Box   $
\end{flushright}
}

\end{document}